\documentclass[12pt]{article}

\usepackage{amssymb}
\usepackage{amsfonts}
\usepackage{amsmath}
\usepackage{amsthm}
\usepackage[english]{babel}
\usepackage{latexsym}
\usepackage{color}
\usepackage{enumerate}
\usepackage{nicefrac}
\usepackage{mathrsfs}
\usepackage{comment}
\usepackage[hmargin=2cm,vmargin=2cm]{geometry}

\usepackage{centernot}

\usepackage[affil-it]{authblk}
\theoremstyle{plain}
\newtheorem{theorem}{Theorem}[section]
\newtheorem{lemma}[theorem]{Lemma}
\newtheorem{proposition}[theorem]{Proposition}
\newtheorem{corollary}[theorem]{Corollary}
\theoremstyle{definition}

\theoremstyle{remark}

\newtheorem{remark}[theorem]{Remark}

\numberwithin{equation}{section}

\newcommand{\abs}[1]{\lvert#1\rvert}
\newcommand{\norm}[1]{\lVert#1\rVert}

\newcommand{\bignorm}[1]{\bigl\lVert#1\bigr\rVert}
\newcommand{\Bigabs}[1]{\Bigl\lvert#1\Bigr\rvert}
\newcommand{\Bignorm}[1]{\Bigl\lVert#1\Bigr\rVert}

\newcommand{\E}{\mathbb{E}}

\newcommand{\probp}{\mathbb{P}}
\newcommand{\probq}{\mathbb{Q}}
\newcommand{\R}{\mathbb{R}}
\newcommand{\N}{\mathbb{N}}
\newcommand{\Q}{\mathbb{Q}}

\newcommand{\cF}{{\mathcal{F}}}

\newcommand{\cM}{\mathcal{M}}

\newcommand{\VaR}{\mathop {\rm VaR}\nolimits}

\newcommand{\ES}{\mathop {\rm ES}\nolimits}

\def\keywords{\vspace{.5em}
{\noindent\textbf{Keywords}:\,\relax%
}}

\def\MSCclassification{\vspace{.5em}
{\noindent\textbf{MSC\,(2000)}:\,\relax%
}}

\makeatletter
\def\@fnsymbol#1{\ensuremath{\ifcase#1\or 1\or 2\or 3\or 4\or 5\or 6\or 7\or 8\else\@ctrerr\fi}}
\makeatother

\begin{document}

\title{Fatou property, representations, and extensions of\\
law-invariant risk measures on general Orlicz spaces}

\author{
\sc{Niushan Gao}\thanks{Department of Mathematics and Computer Science, University of Lethbridge, CA (\texttt{gao.niushan@uleth.ca})}\,, 
\sc{Denny Leung}\thanks{Department of Mathematics, National University of Singapore, SG (\texttt{matlhh@nus.edu.sg})}\,, 
\sc{Cosimo Munari}\thanks{Center for Finance and Insurance, University of Zurich, CH (\texttt{cosimo.munari@bf.uzh.ch})}\,, 
\sc{Foivos Xanthos}\thanks{Department of Mathematics, Ryerson University, CA (\texttt{foivos@ryerson.ca})}
}

\date{\today}

\maketitle

\begin{abstract}\noindent
We provide a variety of results for (quasi)convex, law-invariant
functionals defined on a general Orlicz space, which extend well-known results
in the setting of bounded random variables. First, we show that Delbaen's representation of convex functionals with the Fatou property, which fails in a general
Orlicz space, can be always achieved under the assumption of law-invariance. Second,
we identify the range of Orlicz spaces where the characterization of the Fatou
property in terms of norm lower semicontinuity by Jouini, Schachermayer and
Touzi continues to hold. Third, we extend Kusuoka's representation to a general
Orlicz space. Finally, we prove a version of the extension result by
Filipovi\'{c} and Svindland by replacing norm lower semicontinuity with the
(generally non-equivalent) Fatou property. Our results have natural applications
to the theory of risk measures.
\end{abstract}


\MSCclassification{91B30, 60E05, 46E30, 46A20}


\parindent 0em \noindent


\section{Introduction}

The theory of risk measures is a well-established and still fruitful research
area in the growing field of mathematical finance. In essence, a risk measure
can be viewed as a rule to assign a certain indicator of risk --- typically a
capital requirement --- to a given financial position --- typically the net
capital position (assets net of liabilities) of a financial institution.
Originally articulated in the context of a finite probability space in the
landmark paper by Artzner et al.~\cite{ADEH:99}, the theory was later extended
to general probability spaces by Delbaen \cite{D:02}. In the general case, one
faces the problem of choosing a suitable model for the underlying positions. The
standard theory was developed for bounded positions and a comprehensive account
of the main results in this setting can be found in F\"{o}llmer and Schied
\cite{FS:04}. However, most realistic models in finance and insurance involve
unbounded positions and this calls for a mathematical extension beyond the
bounded setting. A possible extension to the entire set of random variables was
already discussed in Delbaen \cite{D:02} and then again in Delbaen \cite{D:09}. The extension to Lebesgue spaces is presented in Svindland \cite{S:08} and Kaina and R\"{u}schendorf \cite{KR:09},
and the more general extension to the setting of Orlicz spaces was first
investigated by Biagini and Frittelli \cite{BF:10} and Cheridito and Li
\cite{CL:09}. Further results in Orlicz spaces have been obtained by Orihuela and Ruiz Gal\'{a}n \cite{ORG:12}, Kr\"{a}tschmer et al.~\cite{KSZ:14}, Gao and Xanthos \cite{GX:16}, Gao et al.~\cite{GLX:16}, and Delbaen and Owari \cite{DO:16}. A treatment of risk
measures in the context of abstract spaces is provided in Frittelli and Rosazza
Gianin \cite{Fri:02}, Drapeau and Kupper \cite{DK:13}, and Farkas et al.~\cite{FKM:13}.

\medskip

In this paper we work in the context of a nonatomic probability triple
$(\Omega,\cF,\mathbb{P})$ and provide a variety of representation and extension
results for quasiconvex, law-invariant risk measures defined on a {\em general}
Orlicz space $L^\Phi$. In particular, we do not assume that $\Phi$ satisfies the
so-called $\Delta_2$ condition, in which case $L^\Phi$ coincides with its Orlicz
heart $H^\Phi$. The assumption of quasiconvexity is standard in the literature and
reflects the diversification principle according to which the risk of an
aggregated position should be capable of being controlled by the risk of the
individual positions. The assumption of law-invariance, which stipulates that a
risk measure depends solely on the distribution of the underlying position, is
also standard and motivated by the ubiquitous use of time series analysis in
finance and insurance practice.

\medskip

Our main contribution can be
broken down in the following results.

\medskip

\noindent
{\bf Fatou property and dual representations}. It is known since Delbaen
\cite{D:02} that for a proper convex functional
$\rho:L^\infty\to(-\infty,\infty]$ the following are equivalent:
\begin{enumerate}[(1)]
  \item $\rho$ is $\sigma(L^\infty,L^1)$ lower semicontinuous.
  \item $\rho$ has the {\em Fatou property}, i.e.
\[
\mbox{$X_n\to X$ a.s., \ $|X_n|\leq Y$ for some $Y\in L^\infty$} \ \implies \
\rho(X)\leq\liminf_{n\to\infty}\rho(X_n).
\]
\end{enumerate}
In this case, one can always represent $\rho$ in dual terms as follows:
\[
\rho(X) = \sup_{Z\in L^1}\big(\mathbb{E}[ZX]-\rho^\ast(Z)\big), \ \ \ X\in
L^\infty,
\]
where $\rho^\ast:L^1\to(-\infty,\infty]$ is defined by
\[
\rho^\ast(Z) = \sup_{X\in L^\infty}\big(\mathbb{E}[ZX]-\rho(X)\big), \ \ \ Z\in
L^1.
\]
The preceding result shows that, once the Fatou property is fulfilled, the
functional $\rho$ admits a ``nice'' dual representation where the corresponding
dual elements belong to a tractable subspace of the topological dual.
In particular, if $\rho$ is a cash-additive risk measure, the dual elements
can be identified with probability measures that are absolutely continuous with respect to
$\probp$. The appealing feature of the above result is that many risk
measures on $L^\infty$ do satisfy the Fatou property. Most notably, as
established by Jouini et al.~\cite{JST:06}, all convex cash-additive risk
measures on $L^\infty$ that are law-invariant have the Fatou property.

\smallskip

\noindent
It has been an open question since Biagini and Frittelli \cite{BF:10} and Owari \cite{O:14} whether
the above equivalence could be established in the context of a general Orlicz
space $L^\Phi$, where $L^\Phi$ plays the role of $L^\infty$ and $L^\Psi$, with
$\Psi$ being the conjugate of $\Phi$, plays the role of $L^1$. A positive result
was obtained in Delbaen and Owari \cite{DO:16} for a special class of Orlicz spaces. A definitive answer has been finally provided in Gao et al.~\cite{GLX:16}, where the authors proved that the above equivalence holds if, and only if, either the Orlicz function $\Phi$ or its conjugate $\Psi$ is $\Delta_2$.

\smallskip

\noindent
It is therefore natural to wonder whether one can still establish the same
equivalence in the context of a general Orlicz space by imposing suitable
additional assumptions on the underlying functionals. The paper contributes to
this line of research by showing that, under the assumption of {\em
law-invariance}, one can indeed prove the above equivalence without any
restriction on the reference Orlicz space. More specifically, we prove the
following result.

\begin{theorem}\label{repp}
Let $\rho:L^\Phi\rightarrow(-\infty,\infty]$ be a proper, (quasi)convex, law-invariant
functional. Then, the following statements are equivalent:
\begin{enumerate}[(1)]
\item $\rho$ is $\sigma(L^\Phi,L^\infty)$ lower semicontinuous.
\item $\rho$ is $\sigma(L^\Phi,H^\Psi)$ lower semicontinuous.
\item $\rho$ is $\sigma(L^\Phi,L^\Psi)$ lower semicontinuous.
\item $\rho$ has the Fatou property.
\end{enumerate}
\end{theorem}

\medskip

A far-reaching result by Jouini et al.~\cite{JST:06} established that, for a
proper convex functional $\rho:L^\infty\to(-\infty,\infty]$ that is additionally
assumed to be law-invariant, the Fatou property is automatically implied by the (generally weaker) property of norm lower semicontinuity. The result was obtained in the context of a standard nonatomic probability space
and was later extended to arbitrary nonatomic probability spaces in Svindland
\cite{S:10}. Since every cash-additive risk measure on $L^\infty$ is norm
continuous, it follows that a convex cash-additive risk measure on $L^\infty$
has the Fatou property whenever it is law-invariant.

\smallskip

\noindent
We extend the result in Jouini et al.~\cite{JST:06} by characterizing the range
of Orlicz spaces $L^\Phi$ where the above implication remains true for every
proper, convex, law-invariant functional. In particular, we show that norm lower
semicontinuity no longer automatically implies the Fatou property unless $\Phi$
satisfies the $\Delta_2$ condition. This is the content of our second main
result.

\begin{theorem}\label{cha}
The following statements are equivalent:
\begin{enumerate}[(1)]
\item Any proper, (quasi)convex, law-invariant functional
$\rho:L^\Phi\to(-\infty,\infty]$ that is norm lower semicontinuous has the Fatou
property.
\item $\Phi$ is $\Delta_2$.
\end{enumerate}
\end{theorem}

\smallskip

In addition to the above results, we extend to Orlicz spaces the representation
for law-invariant risk measures obtained by Kusuoka \cite{K:01} in the
coherent case and generalized by Frittelli and Rosazza Gianin \cite{FR:05} to
the convex case; see also Shapiro \cite{Sh:13} and Belomestny and Kr\"{a}tschmer \cite{BK:12,BK:17}. Here, we denote by $\ES_\alpha(X)$ the Expected Shortfall of a
random variable $X$ at the level $\alpha\in(0,1]$. Moreover, we denote by
$\mathcal{P}((0,1])$ the set of all probability measures over $(0,1]$.

\begin{theorem}\label{KT}
Let $\rho:L^\Phi\rightarrow(-\infty,\infty]$ be a convex, law-invariant,
cash-additive risk measure with the Fatou property. Then, there exists a proper convex functional
$\gamma:\mathcal{P}((0,1])\rightarrow(-\infty,\infty]$ such that
\[
\rho(X) =
\sup_{\mu\in\mathcal{P}((0,1])}\Big(\int_{(0,1]}
\ES_\alpha(X)\,\mathrm{d}\mu(\alpha)-\gamma(\mu)\Big), \ \ \ X\in L^\Phi.
\]
\end{theorem}

\smallskip

\noindent
We also show that the above Kusuoka representation fails if the Fatou property
is replaced by the weaker property of norm lower semicontinuity.

\medskip

\noindent
{\bf Extensions}. In Filipovi\'{c} and Svindland \cite{FS:12}, it was shown that
every proper, convex, law-invariant, norm lower semicontinuous functional
$\rho:L^\infty\to(-\infty,\infty]$ can be uniquely extended to a convex, law-invariant
functional on $L^p$, $1\leq p<\infty$, that is also norm lower semicontinuous. This extension result played a fundamental role in the study of robustness properties of risk measures as discussed in Kr\"{a}tschmer et al.~\cite{KSZ:14}; see also Koch-Medina and Munari \cite{KM:14}. We show that a similar extension result still holds in the context of a general Orlicz space if one replaces norm lower semicontinuity by the Fatou property.

\begin{theorem}\label{ext}
Any proper, convex, law-invariant functional
$\rho:L^\infty\rightarrow(-\infty,\infty]$ with the Fatou property (or equivalently norm lower semicontinuous) admits a unique proper, convex, law-invariant extension to $L^\Phi$ with the Fatou
property.
\end{theorem}

\smallskip

\noindent  
We also show that $\rho$ may not be extended in a unique way to an Orlicz space if one wants to preserve norm lower semicontinuity only.

\medskip

\noindent
{\bf Structure of the paper}. The paper is structured as follows. In Section
\ref{background}, we recall some fundamental facts about Orlicz spaces and risk
measures. In Section~\ref{ce}, we establish some properties of conditional
expectations on Orlicz spaces. In Section~\ref{li}, we study law-invariant sets
in Orlicz spaces. In Section~\ref{pf}, we provide proofs of the main results
together with some related corollaries.


\section{Orlicz spaces and risk measures}
\label{background}

Throughout the paper we use standard notation from measure theory and functional analysis
as can be found, e.g., in Aliprantis and Border \cite{AB:06}.
We refer to Edgar and Sucheston \cite{ES:02} for a comprehensive account on
Orlicz spaces. A function $\Phi:[0,\infty) \rightarrow[0,\infty)$ is called an
\emph{Orlicz function} if it is convex, increasing, and $\Phi(0)=0$. Define
the \emph{conjugate function} of $\Phi$ by
\[
\Psi(s) = \sup\{ts-\Phi(t) : t \geq 0\}, \ \ \ s\geq0.
\]
If $\lim_{t\to\infty}\frac{\Phi(t)}{t}=\infty$ (or, equivalently, $\Psi$ is
finite-valued), then $\Psi$ is also an Orlicz function, and its conjugate is
$\Phi$. Throughout this paper,  $(\Phi,\Psi)$ stands for a fixed Orlicz pair
satisfying $\Phi(t)>0$ for $t>0$ and
$\lim_{t\to\infty}\frac{\Phi(t)}{t}=\infty$. Note that our restrictions on
$\Phi$ are minor as they only eliminate the case where $L^\Phi$ coincides with
$L^1$ or $L^\infty$, in which cases our main results are either trivial or known.

\smallskip

Fix a nonatomic probability triple $(\Omega,\cF,\mathbb{P})$. In the sequel, we freely use the fact that, for any event $A\in\cF$ and any $p_1,\dots,p_k\geq0$ with $\sum_{i=1}^kp_i\leq\probp(A)$, there exist disjoint measurable subsets $A_1,\dots,A_k$ of $A$ such that $\probp(A_i)=p_i$ for $1\leq i\leq k$; see, e.g., \cite[Section~13.9]{AB:06}.

\smallskip

The \emph{Orlicz space} $L^\Phi:=L^\Phi(\Omega,\cF,\mathbb{P})$ is the Banach lattice of all
random variables $X$ (modulo a.s.~equality under $\probp$) such that
\[
\norm{X}_\Phi := \inf\left\{\lambda>0:\E\left[\Phi\left(\frac{|X|}{\lambda}
\right)\right]\leq 1\right\}<\infty.
\]
The norm $\|\cdot\|_{\Phi}$ is called the \emph{Luxemburg norm}. The subspace of
$L^\Phi$ consisting
of all $X\in L^\Phi $ such that
\[
\E\left[\Phi\left(\frac{|X|}{\lambda}\right)\right]<\infty \ \ \ \mbox{for all
$\lambda>0$}
\]
is conventionally called the \emph{Orlicz heart} of $L^\Phi$ and is denoted by
$H^\Phi$. It is well-known that $L^\infty\subset H^\Phi\subset L^\Phi\subset
L^1$ and that $H^\Phi$ is a norm closed subspace of $L^\Phi$.
Moreover, $L^\Phi=H^\Phi$ if, and only if, the Orlicz function $\Phi$ is
$\Delta_2$, i.e. there exist $t_0\in(0,\infty)$ and $k\in\mathbb{R}$ such that
$
\Phi(2t)<k\Phi(t)
$
for all $t \geq t_0$. We endow the conjugate Orlicz space $L^\Psi$ with
the \emph{Orlicz norm}
\[
\norm{Y}_\Psi := \sup_{X\in L^\Phi, \,\norm{X}_\Phi\leq1}\abs{\mathbb{E}[XY]}, \ \ \ Y\in L^\Psi,
\]
which is equivalent to the Luxemburg norm on $L^\Psi$.
Under the canonical duality $\langle X,Y\rangle:=\E[XY]$ for $X\in L^\Phi$ and $Y\in L^\Psi$, the space $L^\Psi$ can be identified with the \emph{order continuous dual} $(L^\Phi)_n^\sim$ of $L^\Phi$, which is a subspace of the norm dual $(L^\Phi)^*$ of $L^\Phi$. Moreover, $L^\Psi=(L^\Phi)^*$ if, and only if, the function $\Phi$ is $\Delta_2$.

\smallskip

A net $(X_\alpha)$ in $L^\Phi$ is said to \emph{order converge} to $X\in L^\Phi$, denoted $X_\alpha\xrightarrow{o}X$ in $L^\Phi$, if there exists a net $(Y_\alpha)$ in $L^\Phi$ such that $Y_\alpha\downarrow  0$ in $L^\Phi$ and $\abs{X_\alpha-X}\leq Y_\alpha$ for any $\alpha$. For a sequence $(X_n)$ in $L^\Phi$ and $X \in L^\Phi$ one can easily verify that $X_n\xrightarrow{o}X$ is equivalent to dominated almost sure convergence, i.e.~$X_n\rightarrow X$ a.s.~and $\abs{X_n}\leq Y$ for some $Y\in L^\Phi$ and all $n\geq 1$. A set $C\subset L^\Phi$ is \emph{order closed} in $L^\Phi$ if it contains the limit of every order convergent net with elements in $C$. It is well-known that, if $X_\alpha\xrightarrow{o}X$ in $L^\Phi$, then there exists a sequence $(\alpha_n)$ such that $X_{\alpha_n}\xrightarrow{o}X$. Thus, $C$ is order closed in $L^\Phi$ whenever it contains the limit of every order convergent sequence with elements in $C$. Note that every order closed set $C$ is automatically norm closed. Indeed, if $(X_n)\subset C$ converges in norm to $X$, then a subsequence $(X_{n_k})$ order converges to $X$ (see e.g.~\cite[Lemma 3.11]{GX:14}), so that $X\in C$.

\smallskip

A proper (i.e., not identically $\infty$) functional $\rho:L^\Phi\rightarrow(-\infty,\infty]$ is
said to have the {\em Fatou property} whenever
\[
\mbox{$X_n\to X$ a.s., \ $|X_n|\leq Y$ for some $Y\in L^\Phi$ and all $n\in\N$} \ \implies \
\rho(X)\leq\liminf_{n\to\infty}\rho(X_n).
\]
We say that $\rho$ is \emph{order lower semicontinuous} if the sublevel set $\{\rho\leq \lambda\}:=\{X\in L^\Phi : \rho(X) \leq \lambda\}$ is order closed for all $\lambda\in \R$. This is equivalent to
$$
X_n\xrightarrow{o}X\mbox{ in }L^\Phi\implies \liminf_{n\to\infty}\rho(X_n).
$$
In other words, as remarked in Biagini and Frittelli \cite{BF:10}, the Fatou property is equivalent to
order lower semicontinuity. As a result, it follows that a functional with the Fatou property
is automatically norm lower semicontinuous. If $\rho$ is additionally
assumed to be \emph{monotone (decreasing)}, i.e.~$\rho(X)\leq\rho(Y)$ for any
$X,Y\in L^\Phi$ with $X\geq Y$, then the Fatou property is also equivalent to
{\em continuity from above}, i.e.~$\rho(X_n)\to\rho(X)$ whenever $X_n\downarrow X$ in $L^\Phi$.

\smallskip

\noindent
A proper functional $\rho:L^\Phi\to(-\infty,\infty]$ is {\em convex} if $\rho(\lambda X+(1-\lambda)Y)\leq\lambda\rho(X)+(1-\lambda)\rho(Y)$ for any $X,Y\in L^\Phi$ and $\lambda\in[0,1]$ and {\em quasiconvex} if the sublevel set $\{\rho\leq\lambda\}$ is convex for every $\lambda\in\R$. Moreover, we say that $\rho$ is {\em positively homogeneous} if $\rho(\lambda X)=\lambda\rho(X)$ for all $X\in L^\Phi$ and $\lambda\in[0,\infty)$. The functional is called \emph{law-invariant} if $\rho(X)=\rho(Y)$ whenever $X,Y\in L^\Phi$ have the same law.

\smallskip

\noindent
In this paper we use cash-additive risk measures to illustrate our
general results on law-invariant functionals defined on Orlicz spaces. Recall
that $\rho$ is said to be a \emph{cash-additive risk measure} if it is monotone
and satisfies
\[
\rho(X+m1_\Omega)=\rho(X)-m
\]
for any $X\in L^\Phi$ and $m\in\mathbb{R}$. A cash-additive risk measure
is said to be \emph{coherent} if it is convex and positively homogeneous. Two
prominent cash-additive risk measures are the {\em Value-at-Risk} at level
$\alpha\in(0,1)$, which is defined by setting
\[
\VaR_\alpha(X) := \inf\{m\in\mathbb{R} : \mathbb{P}(X+m<0)\leq\alpha\}, \ \ \ X\in L^\Phi,
\]
and the {\em Expected Shortfall} at level $\alpha\in(0,1]$, which is given by
\[
\ES_\alpha(X) := \frac{1}{\alpha}\int_0^\alpha\VaR_\beta(X)\,\mathrm{d}\beta, \ \ \ X\in L^\Phi.
\]
We refer to the above-cited literature for more information about risk measures and their financial applications.


\section{Conditional expectations on Orlicz spaces}\label{ce}

As emerges from Jouini et al. \cite{JST:06} and Svindland \cite{S:10},
conditional expectations played an important role in the study of law-invariant
risk measures on $L^\infty$. However, some key properties of conditional
expectations fail once we abandon the setting of bounded positions. This section
is devoted to collecting a variety of useful properties of conditional
expectations on $L^\Phi$ that will allow us to overcome this failure.

\smallskip

Recall that conditional expectations are contractions on Orlicz spaces, i.e.
\[
\bignorm{\mathbb{E}[X|\mathcal{G}]}_\Phi\leq \norm{X}_\Phi
\]
for any $X\in L^\Phi$ and any $\sigma$-subalgebra $\mathcal{G}$ of $\cF$ (see
\cite[Corollary~2.3.11]{ES:02}). In the sequel, we will write $\pi$ to denote a
finite measurable partition of $\Omega$ whose members have non-zero
probabilities, and denote by $\sigma(\pi)$ the finite $\sigma$-subalgebra
generated by $\pi$. For convenience, we always write
$$
\mathbb{E}[X|\pi] := \mathbb{E}[X|\sigma(\pi)].
$$
The
collection $\Pi$ of all such $\pi$'s is directed by refinement and we write
$\pi'\geq\pi$ whenever $\pi'$ is a refinement of $\pi$. In particular, the
family of conditional expectations $\big(\mathbb{E}[X|\pi]\big)$ becomes a net
with directed set $\Pi$. A fundamental result used in the $L^\infty$-case (see \cite{JST:06} and \cite{S:10}), is recorded in the following lemma.

\begin{lemma}\label{uni-exp}
For any $X\in L^\infty$ we have
\begin{equation}
\label{eq: convergence net cond exp bounded}
\mathbb{E}[X|\pi]\xrightarrow{\norm{\cdot}_\infty}X.
\end{equation}
\end{lemma}

\smallskip

Indeed, for any $\varepsilon>0$, by partitioning $[-\norm{X}_\infty,\norm{X}_\infty]$ into intervals of length at most $\varepsilon$ and considering the corresponding preimages under $X$, we obtain a partition $\pi_0=\{A_1,\dots,A_n\}\in\Pi$ such that the oscillation of $X$ on each $A_i$, $1\leq i\leq n$, is at most $\varepsilon$. Then it is easily seen that $\norm{\E[X|\pi]-X}_\infty\leq \varepsilon$ for all $\pi\geq \pi_0$.
This result, however, fails on Orlicz spaces in
general. Indeed, for $X\in L^\Phi$ we easily see that $\mathbb{E}[X|\pi]\in
L^\infty\subset H^\Phi$ for all $\pi\in\Pi$ and therefore
\[
\mathbb{E}[X|\pi]\xrightarrow{\norm{\cdot}_\Phi}X \ \iff \ X\in H^\Phi.
\]
In particular, condition~\eqref{eq: convergence net cond exp bounded} can be
extended to $L^\Phi$ if, and only if, we have $L^\Phi=H^\Phi$ or, equivalently,
$\Phi$ is $\Delta_2$. The right reformulation of~\eqref{eq: convergence net cond
exp bounded} in the context of a general Orlicz space is as follows.

\begin{proposition}\label{coex}
For every $X\in L^\Phi$ and every $Y\in L^\Psi_+$ we have
\[
\E\big[\abs{\mathbb{E}[X|\pi]-X}Y\big]\rightarrow0.
\]
\end{proposition}
\begin{proof}
Assume first that $X\geq0$. We claim that
\begin{equation}
\label{eq: coex}
\lim_{k\rightarrow\infty}\sup_{\substack{\pi\in\Pi\\A_k\in\sigma(\pi)}}
\E\big[\mathbb{E}[X|\pi]Y1_{A_k}\big]=0,
\end{equation}
where $A_k=\{X>k\}$ for $k\in\N$. To show this, assume that \eqref{eq: coex}
does not hold so that we find $\varepsilon>0$, $k_n\uparrow\infty$, and
$(\pi_n)\subset\Pi$ such that $A_{k_n}\in\sigma(\pi_n)$ for each $n\in\N$
satisfying
\[
\E\big[\mathbb{E}[X|\pi_n]Y1_{A_{k_n}}\big] > \varepsilon \ \ \ \mbox{for all
$n\in\N$}.
\]
Pick any $n\in\N$ and suppose that $A_{k_n}=B_{1,n}\cup\cdots\cup B_{l_n,n}$
with $B_{i,n}\in\pi_n$ for $1\leq i\leq l_n$. Since $\mathbb{P}(A_k)\to0$ as $k\to\infty$ and
$X\in L^1$, one can use Dominated Convergence to infer that $\E[X1_{B_{i,n}\cap A_{k}}]\to0$ as $k\to\infty$, so that $\E[X1_{B_{i,n}\backslash A_{k}}]\to\E[X1_{B_{i,n}}]$ as $k\to\infty$. The same holds for $Y$. Thus, by passing to a convenient subsequence of $(k_n)$, we may assume without loss of generality that each $k_{n+1}$ is large enough so that
\[
\E[X1_{B_{i,n}\backslash A_{k_{n+1}}}]\geq\frac{1}{2}\E[X1_{B_{i,n}}] \ \
\ \mbox{and} \ \ \ 
\E[Y1_{B_{i,n}\backslash A_{k_{n+1}}}]\geq\frac{1}{2}\E[Y1_{B_{i,n}}]
\ \ 
\]for all $1\leq i\leq l_n$.
Write $C_{i,n}=B_{i,n}\backslash A_{k_{n+1}}$ for all $1\leq i\leq l_n$. Then
\[
\sum_{i=1}^{l_n}\E[X|C_{i,n}]\E[Y1_{C_{i,n}}] \geq
\frac{1}{4}\sum_{i=1}^{l_n}\E[X|B_{i,n}]\E[Y1_{B_{i,n}}] =
\frac{1}{4}\E\big[\mathbb{E}[X|\pi_n]Y1_{A_{k_n}}\big] \geq
\frac{\varepsilon}{4}.
\]
Note that $ \mathcal{C}_n=\{C_{1,n},\dots,C_{l_n,n}\}$ is a measurable partition
of $A_{k_n}\backslash A_{k_{n+1}}$ for all $n\in\N$. Thus,
$\{A_{k_1}^c\}\cup\bigcup_n \mathcal{C}_n$ is a measurable partition of
$\Omega$. Let $\mathcal{G}$ be the generated $\sigma$-subalgebra of $\cF$. It
follows from the preceding inequality that
\[
\infty >
\bignorm{\mathbb{E}[X|\mathcal{G}]}_\Phi\norm{Y}_\Psi \geq
\E\big[\mathbb{E}[X|\mathcal{G}]Y\big] \geq
\sum_{n\in\N}\sum_{i=1}^{l_n}\E[X|C_{i,n}]\E[Y1_{C_{i,n}}] =
\infty.
\]
This contradiction completes the proof of \eqref{eq: coex}. As a result, for any
$\varepsilon>0$ there exists $k\in\N$ large enough such that
$\E[\mathbb{E}[X|\pi]Y1_{A_k}]<\varepsilon$ for any $\pi\in\Pi$ with
$A_k\in\sigma(\pi)$. Since $XY\in L^1$, we may take $k$ so large to ensure
$\E[XY1_{A_k}]<\varepsilon$. By the norm convergence of conditional
expectations on $L^\infty$, we find a finite measurable partition $\pi_0\in\Pi$
such that $A_k\in\sigma(\pi_0)$ and
$\bignorm{\mathbb{E}[X1_{A_k^c}|\pi]-X1_{A_k^c}}_\infty<\varepsilon$ for
any $\pi\in\Pi$ refining $\pi_0$. Take now any of such $\pi$'s and note that,
since $A_k\in\sigma(\pi)$, we have
\begin{align*}
\E\big[\abs{\mathbb{E}[X|\pi]-X}Y\big]
=&
\E\big[\abs{\mathbb{E}[X|\pi]-X}Y1_{A_k}\big]+\E\big[\abs{\mathbb{E}[X|\pi]-X
}Y1_{A_k^c}\big] \\
\leq&
\E\big[\mathbb{E}[X|\pi]Y1_{A_k}\big]+\E[XY1_{A_k}]+
\bignorm{\mathbb{E}[X1_{A_k^c}|\pi]-X1_{A_k^c}}_\infty\norm{Y}_1\\
<&
2\varepsilon+\varepsilon\norm{Y}_1.
\end{align*}
This establishes the assertion for $X\geq0$. To conclude, take now an
arbitrary $X\in L^\Phi$ and note that
\[
\E\big[\abs{\mathbb{E}[X|\pi]-X}Y\big] \leq
\E\big[\abs{\mathbb{E}[X^+|\pi]-X^+}Y\big]+\E\big[\abs{\mathbb{E}[X^-|\pi]-X^-}
Y\big]
\rightarrow 0
\]
by what we have just established.
\end{proof}

\smallskip

In light of the link between the Fatou property and order lower semicontinuity,
order convergence of conditional expectations is most desired.  But it also
fails on Orlicz spaces in general.

\begin{proposition}\label{not-o}
Let $X\in L^\Phi$. Then, $\mathbb{E}[X|\pi]\xrightarrow{o}X$ in $L^\Phi$ if and
only if $X\in L^\infty$.
\end{proposition}
\begin{proof}
Suppose that $X\in L^\infty$. For each $\pi\in\Pi$ set $\lambda_\pi=\sup_{\pi'\geq \pi}\norm{\E[X|\pi']-X}_\infty$ and note that $\lambda_\pi\downarrow0$ by Lemma~\ref{uni-exp}. Then, $\abs{E[X|\pi]-X}\leq\lambda_\pi1_\Omega$ implies $\mathbb{E}[X|\pi]\xrightarrow{o}X$ in $L^\Phi$.

Conversely, suppose that $X \notin L^\infty$. Without loss of generality, assume
that $\probp(X>k)>0$ for all $k\in\N$ and let $\pi_0=\{A_1,\dots, A_n\}$ be an
arbitrary member of $\Pi$. It is easy to see that there
exists some $A_i$, $1\leq i\leq n$, such that $\probp(A_i\cap\{X>k\})>0$ for all
$k\in\N$. Say, $A_1$ is as such. Fix now an arbitrary $k\in\N$. If
$\probp(A_1\cap\{X\leq k\})=0$, then we immediately see that
\[
\mathbb{E}[X|\pi_0]\geq k \ \ \ \mbox{a.s.~on $A_1$}.
\]
Otherwise, we must have $\probp(A_1\backslash\{X>k\})>0$. Set $c=\probp(A_1\cap\{X>k\})$. Since $X\in
L^1$, there exists $0<\varepsilon<c$ such that $\E[\abs{X}1_{B}]<\frac{kc}{2}$
whenever $B\in\cF$ satisfies $\probp(B)\leq\varepsilon$. 
By nonatomicity, we can take finitely many measurable subsets $B_1,\dots, B_j$ of $A_1\backslash\{X>k\}$ such that $\probp(B_i)<\varepsilon$ for each $1\leq i\leq j$ and $\bigcup_{i=1}^jB_i=A_1\backslash \{X>k\}$.
Now, for any $1\leq i\leq j$, the set $C_i=(A_1\cap\{X>k\})\cup B_i\subset A_1$ satisfies $\probp(C_i)\leq 2c$. Take a refinement $\pi_i\geq \pi_0$ such that $C_i\in \pi_i$.
Then, we have
\begin{align*}
\sup_{\pi\geq \pi_0}\E[X|\pi]\geq& \E[X|\pi_i] =\frac{1}{\probp(C_i)}\E[X1_{A_1\cap\{X>k\}}+X1_{B_i}]\\
\geq&
\frac{1}{2c}\big(\E[X1_{A_1\cap\{X>k\}}]-\E[\abs{X}1_{B_i}]\big)
\geq
\frac{1}{2c}\Big(kc-\frac{kc}{2}\Big) = \frac{k}{4} \ \ \ \mbox{a.s.~on $C_i$}.
\end{align*}
Since $\bigcup_{i=1}^jC_i=A_1$, we infer that
\[
\sup_{\pi\geq \pi_0}\mathbb{E}[X|\pi] \geq \frac{k}{4} \ \ \ \mbox{a.s.~on $A_1$}.
\]
Since $k$ was arbitrary, it follows that
\[
\sup_{\pi\geq\pi_0}\mathbb{E}[X|\pi]=\infty \ \ \ \mbox{a.s.~on $A_1$}.
\]
This implies that the net $(\mathbb{E}[X|\pi])$ has no order bounded tail and,
thus, does not order converge. In particular, $(\mathbb{E}[X|\pi])$ does not order converge to $X$.
\end{proof}

\smallskip

In spite of the preceding negative result, for every random variable $X\in L^\Phi$ we can always select a sequence of partitions with respect to which the conditional expectations of $X$ do order
converge to $X$ itself.

\begin{proposition}\label{ce-con}
For any $X\in L^\Phi$ there exists a sequence $(\pi_n)\subset\Pi$ such that
\[
\mbox{$\mathbb{E}[X|\pi_n]\xrightarrow{o} X$ in $L^\Phi$}.
\]
\end{proposition}
\begin{proof}
Without loss of generality assume that $\norm{X}_\Phi\leq \frac{1}{2}$ so that
$\E[\Phi(2\abs{X})]<\infty$. For each $n\in\N$ take $k_n\in\N$ large enough such
that
\begin{equation}
\label{eq: ce-con}
\E[\Phi(2\abs{X})1_{A_n}] \leq \frac{1}{2^n}
\end{equation}
where $A_n=\{\abs{X}\geq k_n\}$. Then, take $\pi_n\in\Pi$ satisfying
$A_n\in\sigma(\pi_n)$ and
\begin{equation}
\label{eq: ce-con 1}
\bignorm{X1_{A_n^c}-\mathbb{E}[X1_{A_n^c}|\pi_n]}_\infty\leq\frac{1}{2^n}.
\end{equation}
We claim that $\mathbb{E}[X|\pi_n]\xrightarrow{o} X$ in $L^\Phi$. To prove this,
fix $n\in\N$ and note first that
\begin{equation}
\label{eq: ce-con 2}
\E\big[\Phi(\E[2\abs{X}|A_n])1_{A_n}\big] \leq
\E\big[\E[\Phi(2\abs{X})|A_n]1_{A_n}\big] =
\E[\Phi(2\abs{X})1_{A_n}] \leq
\frac{1}{2^n}
\end{equation}
by the conditional version of Jensen's Inequality. Moreover, assumption
\eqref{eq: ce-con 1} ensures that
\begin{equation}
\label{eq: ce-con 3}
\bignorm{X1_{A_n^c}-\mathbb{E}[X1_{A_n^c}|\pi_n]}_\Phi \leq
\frac{\norm{1_\Omega}_\Phi}{2^n}.
\end{equation}
Finally, set
\[
Y_n=X1_{A_n}-\mathbb{E}[X1_{A_n}|\pi_n] \ \ \ \mbox{and} \ \ \
Z_n=X1_{A_n^c}-\mathbb{E}[X1_{A_n^c}|\pi_n]
\]
and note that
\[
X-\mathbb{E}[X|\pi_n]=Y_n+Z_n.
\]
It is clear that, setting
\[
X_0=\sup_{n\in\N}\abs{Y_n}+\sum_{n\in\N}\abs{Z_n}\,,
\]
we have $\abs{X-\mathbb{E}[X|\pi_n]}\leq X_0$ for all $n\in\N$. We claim that
$X_0\in L^\Phi$. To show this, note first that
\[
\norm{Z_n}_\Phi\leq\frac{\norm{1_\Omega}_\Phi}{2^n}
\]
by \eqref{eq: ce-con 3}. Hence, $(\sum_{n=1}^m\abs{Z_n})_m$ is a Cauchy sequence with respect to $\norm{\cdot}_\Phi$, and thus $\sum_{n\in\N}\abs{Z_n}\in L^\Phi$. By
\eqref{eq: ce-con} and \eqref{eq: ce-con 2} and by convexity of $\Phi$, we have
\[
\E[\Phi(\abs{Y_n})] \leq
\frac{1}{2}\E[\Phi(2\abs{X})1_{A_n}]+\frac{1}{2}\E\big[\Phi(\E[2\abs{X}|A_n]
)1_{A_n}\big] \leq
\frac{1}{2^n}
\]
for every $n\in\N$. Hence, the continuity and strict monotonicity of $\Phi$
yield
\[
\E\Big[\Phi\Big(\sup_{n\in\N}\abs{Y_n}\Big)\Big] =
\E\Big[\sup_{n\in\N}\Phi(\abs{Y_n})\Big] \leq
\sum_{n\in\N}\E[\Phi(\abs{Y_n})] \leq
\sum_{n\in\N}\frac{1}{2^n} = 1.
\]
This shows that $\sup_{n\in\N}\abs{Y_n}\in L^\Phi$ as well, so that $X_0\in
L^\Phi$. Now, by Markov's Inequality, we have
\[
\Phi(\varepsilon)\probp(\abs{Y_n}>\varepsilon) \leq \E[\Phi(\abs{Y_n})] \leq
\frac{1}{2^n}
\]
for any $\varepsilon>0$ and $n\in\N$. It follows that $(Y_n)$ converges to $0$
in probability. Since $\norm{Z_n}_\Phi \to 0$, \cite[Corollary~2.1.10]{ES:02}
implies that $(Z_n)$ also converges to $0$ in probability. As a result, the
sequence $(X-\mathbb{E}[X|\pi_n])$ converges to $0$ in probability. A
subsequence of it converges to $0$ a.s., and thus in order, since even the whole
sequence $(X-\mathbb{E}[X|\pi_n])$ is order bounded by $X_0$.
The sequence of partitions corresponding to this special subsequence, which we still
denote by $(\pi_n)$, is then easily seen to satisfy $\E[X|\pi_n]\xrightarrow{o} X$ in $L^\Phi$.
\end{proof}


\section{Law-invariant sets in Orlicz spaces}\label{li}

In this section we establish a key result on law-invariant sets, which will be
later applied to level sets of law-invariant functionals defined on Orlicz
spaces. Here, we say that a subset $C$ of $L^\Phi$ is \emph{law-invariant} if
$X\in C$ for any $X\in L^\Phi$ that has the same law of some element of $C$.

\smallskip

We start with the following observation. For any $A\in\cF$ with $\probp(A)>0$ consider the ``localized'' probability space $(A,\cF_{|A},\probp_{|A})$, where we set $\cF_{|A}:=\{B\in\cF : B\subset A\}$ and
$\probp_{|A}:\cF_{|A}\to[0,1]$ is defined by $\probp_{|A}(B):=\probp(B|A)$. Note that $(A,\cF_{|A},\probp_{|A})$ is also nonatomic. For any $X\in L^\Phi$ we denote by $X_{|A}$ the random variable on $(A,\cF_{|A},\probp_{|A})$ obtained by restricting $X$ to $A$. Let now $(\Omega',\cF',\mathbb{P}')$ be any nonatomic probability space and recall that, applying any quantile function of $X$ to a uniform random variable over $(\Omega',\cF',\mathbb{P}')$, we obtain a random variable over $(\Omega',\cF',\mathbb{P}')$ that has the same law as $X$. Working at a ``localized'' level, we can use the same idea to show that, given two sets $A,B\in\cF$ with $\probp(A)=\probp(B)$, we always find a random variable $Z\in L^\Phi$ such that $X1_A$ has the same law as $Z1_B$.

\begin{lemma}\label{critical}
Let $X\in L^\Phi$, $\varepsilon >0$ and $\pi\in\Pi$ be fixed. Then, there exist
$B\in\cF$ with $\mathbb{P}(B)\leq\varepsilon$ and $X_1,\dots,X_N\in L^\Phi$ with
the same law as $X$ such that, setting $Y=\frac{1}{N}\sum_{i=1}^NX_i$, we have
\[
\mathbb{E}[Y|\pi]=\mathbb{E}[X|\pi], \ \ \norm{Y1_B}_\Phi\leq\varepsilon, \ \
Y1_{B^c}\in L^\infty.
\]
\end{lemma}
\begin{proof}
Without loss of generality, we assume
that $\varepsilon<1$ and $\norm{X}_\Phi \leq 1$ so that $\E[\Phi(\abs{X})]\leq
1$.

\smallskip

Let $A\in\cF$ be such that $\probp(A)>0$.  Choose $N\in \N$ such that $N\geq\frac{2}{\varepsilon}$ and $c>0$ such that $\mathbb{P}(\{\abs{X}<c\}\cap A)>0$. Moreover, choose $c'>0$ large enough to ensure that
\[
c'>(N-1)c, \ \
\probp(\{\abs{X}>c'\}\cap A)\leq\frac{1}{N}\min\big(\mathbb{P}(\{\abs{X}<c\}\cap A),\varepsilon\big), \
\ \E[\Phi(\abs{X})1_{\{\abs{X}>c'\}}]\leq\frac{1}{N}.
\]
Set $C = \{|X| > c'\}\cap A$ and note that $N\probp(C) \leq \probp(\{|X| <c\}\cap A)$.
By nonatomicity, there exist pairwise disjoint measurable subsets $B_1,\dots, B_N$
of $\{|X|<c\}\cap A$ such that $\probp(B_i)=\probp(C)$ for all $1\leq i\leq N$. Observe that $B_i$ is disjoint from $C$ for all $1\leq i\leq N$ because $c'>c$. Now, for any fixed $i\in\{1,\dots,N\}$ we may use the above observation to ensure the existence of $Z_i,Z'_i\in L^\Phi$ such that $Z_i1_{B_i}$ has the same law as $X1_C$ and $Z_i'1_C$ has the same law as $X1_{B_i}$. Setting $X_i=Z_i1_{B_i}+Z_i'1_C+X1_{D_i}$, where $D_i=A\backslash(C\cup B_i)$, one clearly sees that $X_i=X_i1_A$ has the same law as $X1_A$. Define now $Y=\frac{1}{N}\sum^N_{i=1}X_i$. Moreover, set $B=\bigcup^N_{i=1}B_i\subset A$ and note that $\probp(B)\leq\varepsilon$. Since for each $1\leq i\leq n$ the random variable $Z'_i1_C$ has the same law as $X1_{B_i}$ and $\abs{X}<c$ a.s.~on $B_i$, we see that $\abs{Z'_i}<c$ a.s.~on $C$. This, together with the inclusion $D_i\subset A\backslash C\subset\{\abs{X}\leq c'\}$, implies that
\[
\abs{X_i1_{B^c}} = \abs{Z'_i1_{C\cap B^c}}+\abs{X1_{D_i\cap B^c}} \leq (c+c')1_\Omega.
\]
This shows that $Y1_{B^c}\in L^\infty$. Next, we claim that $\|Y1_B\|_\Phi \leq \varepsilon$. To prove this, fix $j\in\{1,\dots,N\}$ and note first that $\abs{X_j1_{B_i}}=\abs{X1_{D_j\cap B_i}}\leq c1_{B_i}$ for all $1\leq i\leq N$ with $i\neq j$. In addition, since $Z_j1_{B_j}$ has the same law as $X1_C$ and $\abs{X}>c'$ a.s.~on $C$, we easily see that $\abs{X_j}=\abs{Z_j}>c'$ a.s.~on $B_j$. As a result, we obtain
\begin{align*} 
N\abs{Y1_{B_j}}
\leq&
\sum_{\substack{i=1\\i\neq j}}^N\abs{X_i1_{B_j}} + \abs{X_j1_{B_j}} \leq {(N-1)c1_{B_j}+ \abs{X_j1_{B_j}}} \\
\leq&
c'1_{B_j} + \abs{X_j1_{B_j}} \leq {2}\abs{X_j1_{B_j}}.
\end{align*}
Since $\abs{X_i}>c'$ a.s.~on $B_i$ and $X_i$ has the same law as $X1_A$ for all $1\leq i\leq N$, it follows that
\begin{align*}
\E\bigg[\Phi\Big(\frac{N}{2}\abs{Y}\Big)1_B\bigg]
=&
\sum^N_{i=1}\E\bigg[\Phi\Big(\frac{N}{2}\abs{Y}\Big)1_{B_i}\bigg] 
\leq
\sum^N_{i=1}\E[\Phi(\abs{X_i})1_{B_i}] 
\leq\sum^N_{i=1}\E[\Phi(\abs{X_i})1_{\{\abs{X_i}>c'\}}]\\
=&N\E[\Phi(\abs{X1_A})1_{\{\abs{X1_A}>c'\}}]
\leq
N\E[\Phi(\abs{X})1_{\{\abs{X}>c'\}}] \leq 1.
\end{align*}
This yields $\norm{Y1_B}_\Phi \leq \frac{2}{N}\leq \varepsilon$ and concludes the proof of the claim. 

Suppose now that $\pi=\{A_1,\dots, A_n\}$ and fix $k\in\{1,\dots,n\}$. Applying the preceding argument to $A_k$, we can ensure the existence of a measurable set $B_k\subset A_k$ with $\mathbb{P}(B_k)\leq\frac{\varepsilon}{n}$ as well as of random variables $X_{ki}\in L^\Phi$, $1\leq
i\leq N_k$, such that $X_{ki}$ is zero a.s.~on $A_k^c$ and has the same law as $X1_{A_k}$ and such that, setting $Y_k=\frac{1}{N_k}\sum_{i=1}^{N_k}X_{ki}$, we have
\[
\norm{Y_k1_{B_k}}_\Phi \leq \frac{\varepsilon}{n} \ \ \ \mbox{and} \ \ \
Y_k1_{B_k^c}\in L^\infty.
\]
Now, set $Y=\sum_{k=1}^nY_k$ and $B=\bigcup_{k=1}^nB_k$. Then
$\mathbb{P}(B)\leq\varepsilon$ and $Y1_{B^c}\in L^\infty$. Moreover, since $Y_k$ vanishes on $A_j$ for $j\neq k$, we have $Y1_B=\sum_{k=1}^nY_k1_{B_k}$ and, hence, $\norm{Y1_B}_\Phi
\leq \varepsilon$. It also follows that
$\E[Y1_{A_k}]=\E[Y_k1_{A_k}]=\E[X1_{A_k}]$ for every $1\leq k\leq n$,
so that $\mathbb{E}[Y|\pi]=\mathbb{E}[X|\pi]$. Setting
$N=\prod_{k=1}^nN_k$ and, for every $1\leq k\leq n$, repeating each $X_{ki}$,
$1\leq i\leq N_k$, for $\frac{N}{N_k}$ times, we can write
$Y_k=\frac{1}{N}\sum_{i=1}^NX_i^{(k)}$, where each $X_i^{(k)}$ is zero a.s.~on $A_k^c$ and has the same law as $X1_{A_k}$. Then
$Y=\frac{1}{N}\sum_{i=1}^N\sum_{k=1}^nX_i^{(k)}$ and, clearly, each
$\sum_{k=1}^nX_i^{(k)}$ has the same law as $X$.
\end{proof}

\smallskip

To establish our key result on law-invariant sets we also need to use the
following result, which is contained in Step 2 in the proof of
\cite[Lemma~1.3]{S:10}.

\begin{lemma}\label{bounded}
Let $X\in L^\infty$, $\varepsilon>0$ and $\pi\in\Pi$. Then, there exist
$X_1,\dots,X_N\in L^\infty$ which have the same law as $X$ and satisfy
\[
\Bignorm{\frac{1}{N}\sum_{i=1}^NX_i-\mathbb{E}[X|\pi]}_\infty \leq \varepsilon.
\]
\end{lemma}

\begin{proposition}\label{con-con-e}
Let $C$ be a convex, norm closed, law-invariant set in $L^\Phi$. Then,
$\mathbb{E}[X|\pi]\in C$ for any $X\in C$ and any $\pi\in\Pi$.
\end{proposition}
\begin{proof}
Let $X\in C$, $\pi=\{A_1,\dots,A_n\}\in\Pi$ and fix an arbitrary
$\varepsilon>0$. For any $B\in\cF$ and $Y\in L^\Phi$ we have
\begin{align*}
&\,\bignorm{\E[Y|A_i\cap B^c]1_{A_i\cap B^c}-\E[Y|A_i]1_{A_i}}_\Phi \\
\leq&\,
\bignorm{(\E[Y|A_i\cap B^c]-\E[Y|A_i])1_{A_i\cap B^c}}_\Phi +
\bignorm{\E[Y|A_i]1_{A_i\cap B}}_\Phi \\
=&\,
\bignorm{\probp(B|A_i)\big(\E[Y|A_i\cap B^c]-\E[Y|A_i\cap B]\big)1_{A_i\cap B^c}}_\Phi + \bignorm{\E[Y|A_i]1_{A_i\cap B}}_\Phi \\
\leq&\,
\bigg(\frac{\probp(B)\E[\abs{Y}]}{\probp(A_i)\probp(A_i\cap B^c)}+
\frac{\norm{Y1_B}_\Phi\norm{1_\Omega}_\Psi}{\probp(A_i)}\bigg)
\norm{1_\Omega}_\Phi+
\frac{\E[\abs{Y}]}{\probp(A_i)}\norm{1_B}_\Phi
\end{align*}
for every $1\leq i\leq n$. As a result, there exists $0<\delta<\varepsilon$ such
that, whenever $\probp(B)<\delta$, $\norm{Y1_B}_\Phi<\delta$ and
$\E[\abs{Y}]\leq\E[\abs{X}]$, we have
\begin{equation}
\label{eq: con-con-e}
\Bignorm{\sum_{i=1}^n\E[Y|A_i\cap B^c]1_{A_i\cap B^c}-\E[Y|\pi]}_\Phi <
\varepsilon.
\end{equation}
Here, we have used the fact that
$\norm{1_A}_\Phi\to0$ as $\probp(A)\to0$. 
Applying Lemma~\ref{critical}, we obtain $B\in\cF$ with $\probp(B)<\delta$ and
$X_1,\dots,X_N\in L^\Phi$ which have the same law as $X$ and satisfy
\[
\mathbb{E}[Y|\pi]=\mathbb{E}[X|\pi], \ \ \ \norm{Y1_B}_\Phi<\delta, \ \ \
Y1_{B^c}\in L^\infty,
\]
where we set $Y=\frac{1}{N}\sum_{i=1}^NX_i$. In particular, note that $Y\in C$
by convexity and law-invariance of $C$. Moreover, it follows from
$\abs{Y}\leq\frac{1}{N}\sum_1^N\abs{X_i}$ that $\E[\abs{Y}]\leq\E[\abs{X}]$, so
that \eqref{eq: con-con-e} holds. 
Now, consider the nonatomic probability space $(B^c,\cF_{|B^c},\probp_{|B^c})$. Applying Lemma~\ref{bounded} to $Y_{|B^c}$ and the partition $\{A_1\cap B^c,\dots,A_n\cap B^c\}$ of the state space $B^c$, we obtain $Y_1',\dots,Y_M'\in L^\infty(B^c,\cF_{|B^c},\probp_{|B^c})$ such that $Y'_j$ has the same law as $Y_{|B^c}$ for all $1\leq j\leq M$ and 
\[
\Bigabs{\sum_{i=1}^n\E_{|B^c}[Y_{|B^c}|A_i\cap B^c]1_{A_i\cap
B^c}-\frac{1}{M}\sum_{j=1}^MY_j'}\leq \varepsilon\quad \mbox{$\probp_{|B^c}$-a.s.~on $B^c$},
\]
where $\E_{|B^c}$ denotes the expectation under $\probp_{|B^c}$. A direct computation shows that $\E[Y|A_i\cap B^c]=\E_{|B^c}[Y_{|B^c}|A_i\cap B^c]$ for all $1\leq i\leq n$, so that
\[
\Bigabs{\sum_{i=1}^n\E[Y|A_i\cap B^c]1_{A_i\cap
B^c}-\frac{1}{M}\sum_{j=1}^MY_j'}\leq \varepsilon\quad \mbox{$\probp_{|B^c}$-a.s.~on $B^c$}.
\]
Set $Y_j=Y1_B+Y_j'1_{B^c}$ for $1\leq j\leq M$. Then, $Y_j$ has the same law as
$Y$ and, hence, $Y_j\in C$ by law-invariance, for every $1\leq j\leq M$. Note that
$\frac{1}{M}\sum_{j=1}^MY_j\in C$ by convexity of $C$ and that
\begin{align}
\label{eq: con-con-e 1}
\nonumber &\Bignorm{\sum_{i=1}^n\E[Y|A_i\cap B^c]1_{A_i\cap
B^c}-\frac{1}{M}\sum_{j=1}^MY_j}_\Phi\\
=&\Bignorm{\sum_{i=1}^n\E[Y|A_i\cap B^c]1_{A_i\cap
B^c}-\frac{1}{M}\sum_{j=1}^MY_j'1_{B^c}-Y1_B}_\Phi\\
\nonumber \leq& \varepsilon\norm{1_{B^c}}_\Phi+\norm{Y1_B}_\Phi\leq
\varepsilon\norm{1_\Omega}_\Phi+\varepsilon.
\end{align}
Since $\E[Y|\pi]=\E[X|\pi]$, we can easily combine \eqref{eq: con-con-e} and
\eqref{eq: con-con-e 1} to obtain
\[
\Bignorm{\E[X|\pi]-\frac{1}{M}\sum_{j=1}^MY_j}_\Phi \leq
(2+\norm{1_\Omega}_\Phi)\varepsilon.
\]
Finally, by norm closedness of $C$, we infer that $\mathbb{E}[X|\pi]\in C$.
\end{proof}

\begin{remark}
Note that, in order to obtain the above proposition, one cannot directly
truncate an arbitrary $X\in L^\Phi$ and then apply Lemma~\ref{bounded} because
the tail of $X$ may not have small norm. In fact,
$\norm{X1_{\{\abs{X}>k\}}}_\Phi\to0$ as $k\to\infty$ precisely when
$X\in H^\Phi$.
\end{remark}

\smallskip

Recall that every order closed set in $L^\Phi$ is automatically norm closed. This, together with Proposition~\ref{ce-con} and
Proposition~\ref{con-con-e}, immediately implies the following characterization
of the elements of a law-invariant set in terms of their conditional
expectations.

\begin{corollary}\label{cha-order}
Let $C$ be a convex, law-invariant, order closed set in $L^\Phi$. Then, for
every $X\in L^\Phi$ we have $X\in C$ if, and only if, $\E[X|\pi]\in C$ for any $\pi\in\Pi$.
\end{corollary}

\smallskip

We are now in a position to derive the main result of this section, which shows
that, for a law-invariant set in $L^\Phi$, order closedness and
$\sigma(L^\Phi,L^\Psi)$ closedness are equivalent. When applied to level sets of
convex, law-invariant functionals on $L^\Phi$, this equivalence will immediately
yield the desired characterization of the Fatou property in terms of
$\sigma(L^\Phi,L^\Psi)$ lower semicontinuity.

\begin{corollary}\label{rep}
For a convex law-invariant set $C$ in $L^\Phi$ the following are equivalent:
\begin{enumerate}[(1)]
\item $C$ is order closed.
\item $C$ is $\sigma(L^\Phi,L^\Psi)$ closed.
\item $C$ is $\sigma(L^\Phi,H^\Psi)$ closed.
\item $C$ is $\sigma(L^\Phi,L^\infty)$ closed.
\end{enumerate}
\end{corollary}
\begin{proof}
Clearly, we only need to prove that (1) implies (4). To this effect, assume that
$C$ is order closed and consider a net $(X_\alpha)\subset C$ and $X\in L^\Phi$
such that
\[
X_\alpha\xrightarrow{\sigma(L^\Phi,L^\infty)}X.
\]
Moreover, fix a partition $\pi=\{A_1,\dots,A_n\}\in\Pi$. Then, for any norm continuous linear functional $\phi:L^\Phi\rightarrow\R$, we have
\[
\phi\big(\E[X_\alpha|\pi]\big) =
\E\Big[X_\alpha\sum_{i=1}^n\frac{\phi(1_{A_i})}{\probp(A_i)}1_{A_i}\Big] \to
\E\Big[X\sum_{i=1}^n\frac{\phi(1_{A_i})}{\probp(A_i)}1_{A_i}\Big] =
\phi\big(\E[X|\pi]\big).
\]

Thus, $(\E[X_\alpha|\pi]) $ converges weakly to $\E[X|\pi]$. Since
$\mathbb{E}[X_\alpha|\pi]\in C$ for any index $\alpha$ by
Corollary~\ref{cha-order} and since, being order closed and thus norm closed, the convex set
$C$ is weakly closed, we infer that $\mathbb{E}[X|\pi]\in C$. In light of
Corollary~\ref{cha-order}, this yields $X\in C$ and proves that $C$ is
$\sigma(L^\Phi,L^\infty)$ closed.
\end{proof}


\section{Proofs of the Main Results}\label{pf}

In this final section we prove the results stated in the
introduction and derive a
variety of corollaries for functionals and risk measures defined on Orlicz
spaces.

\begin{proof}[{\bf{Proof of Theorem~\ref{repp}}}]
It is straightforward to verify that $\rho$ is order lower semicontinuous if,
and only if, the level set
$\{\rho\leq \lambda\}$ is order closed for any $\lambda\in\mathbb{R}$. Recall
also that $\rho$ is
$\sigma(L^\Phi,L^\Psi)$ (respectively, $\sigma(L^\Phi,H^\Psi)$ and
$\sigma(L^\Phi,L^\infty)$) lower semicontinuous if, and only if, each level set
is $\sigma(L^\Phi,L^\Psi)$ (respectively, $\sigma(L^\Phi,H^\Psi)$ and
$\sigma(L^\Phi,L^\infty)$) closed. Since each level set is convex and
law-invariant, the equivalence follows directly from Corollary~\ref{rep}.
\end{proof}

\smallskip

\noindent
The following dual representation of functionals with the Fatou 
property is an immediate consequence of the above theorem and of Fenchel-Moreau duality.

\begin{corollary}
Let $\rho:L^\Phi\rightarrow(-\infty,\infty]$ be a proper, 
convex, law-invariant
functional with the Fatou property. Then, we have
\[
\rho(X) = \sup_{Z\in L^\Psi}\big(\E[ZX]-\rho^\ast(Z)\big), 
\ \ \ X\in L^\Phi,
\]
where
\[
\rho^\ast(Z) = \sup_{X\in L^\Phi}\big(\E[ZX]-\rho(X)\big), 
\ \ \ Z\in L^\Psi.
\]
The first supremum can be equivalently taken 
over $H^\Psi$ or $L^\infty$.
\end{corollary}

\smallskip

\noindent
We specify the preceding theorem to cash-additive risk measures. Here, we
denote by $\cM(L^\Psi)$ (respectively, $\cM(H^\Psi)$ and $\cM(L^\infty)$) the
set of all probability measures $\Q$ over $\Omega$ that are absolutely
continuous with respect to $\probq$ and such that $\frac{d\Q}{d\probp}$ belongs to
$L^\Psi$ (respectively, $H^\Psi$ and $L^\infty$).

\begin{corollary}
Let $\rho:L^\Phi\rightarrow(-\infty,\infty]$ be a proper, convex, law-invariant,
cash-additive risk measure with the Fatou property. Then, we have
\[
\rho(X) = \sup_{\Q\in\cM(L^\Psi)}\big(\E_\Q[-X]-\rho^\ast(\Q)\big), \ \ \ X\in
L^\Phi,
\]
where
\[
\rho^\ast(\Q) = \sup_{X\in L^\Phi}\big(\E_\Q[-X]-\rho(X)\big), \ \ \
\Q\in\cM(L^\Psi).
\]
The first supremum can be equivalently taken over $\cM(H^\Psi)$ or
$\cM(L^\infty)$.
\end{corollary}

\smallskip

\begin{proof}[{\bf{Proof of Theorem~\ref{cha}}}]
First, assume that $\Phi$ is $\Delta_2$. Then, $L^\Phi=H^\Phi$ and every order
convergent sequence is also norm convergent to the same limit by
\cite[Theorem~2.1.14]{ES:02}). In particular, every norm closed set is also
order closed. Let $\rho:L^\Phi\to(-\infty,\infty]$ be a proper, (quasi)convex,
law-invariant functional that is norm lower semicontinuous. Since every level
set of $\rho$ is norm closed and, hence, order closed, it follows that $\rho$ is
order lower semicontinuous or, equivalently, has the Fatou property. This shows
that (2) implies (1).

\smallskip

To prove the converse implication, assume that $\Phi$ is not $\Delta_2$ so that
$L^\Phi\neq H^\Phi$ and set
\[
C = \{X\in L^\Phi : X^-\in H^\Phi, \ \E[X]\geq 0\}.
\]
It is clear that $C$ is a law-invariant cone. Moreover, it is convex since for
any $X,Y\in C$ and $\lambda\in[0,1]$ we have
\[
0 \leq (\lambda X+(1-\lambda)Y)^- \leq \lambda X^-+(1-\lambda)Y^- \in H^\Phi,
\]
showing that $(\lambda X+(1-\lambda)Y)^-\in H^\Phi$. The set $C$ is also easily
seen to be monotone. Indeed, for any $X\in C$ and $Y\in L^\Phi$ with $Y\geq X$
it holds
\[
0 \leq Y^- \leq X^- \in H^\Phi,
\]
implying that $Y\in C$. We claim that $C$ is norm closed but not order closed.
To establish norm closedness, take $(X_n)\subset C $ and $X\in L^\Phi$
satisfying $X_n\xrightarrow{\norm{\cdot}_\Phi}X$. Since this implies
$X_n\xrightarrow{\norm{\cdot}_1}X$, we see that $\E[X]\geq0$. Moreover, since
$H^\Phi$ is norm closed, it follows from
$\norm{X_n^--X^-}_\Phi\leq\norm{X_n-X}_\Phi\rightarrow0$ that $X^-\in H^\Phi$.
Therefore, we conclude that $X\in C$, showing that $C$ is norm closed.

\smallskip

To prove that $C$ is not order closed, take a positive non-zero $Y\in
L^\Phi\backslash H^\Phi$. Moreover, take any $\lambda\geq\mathbb{E}[Y]$ and set
\[
X_n=\lambda1_\Omega-\min(Y,n1_\Omega), \ n\in\N, \ \ \ \mbox{and} \ \ \ X=\lambda1_\Omega-Y.
\]
Note that $(X_n)\subset L^\infty\subset H^\Phi$ so that $(X_n^-)\subset H^\Phi$
and $\E[X_n]\geq\lambda-\E[Y]\geq0$ for every $n\in\N$, showing that
$(X_n)\subset C$. Clearly, we have $X_n\downarrow X$ so that
$X_n\xrightarrow{o}X$ in $L^\Phi$. However, $X$ does not belong to $C$, for
otherwise $X^-\in H^\Phi$ and $0\leq X^+\leq\lambda1_\Omega\in H^\Phi$ would imply
$X\in H^\Phi$ and, thus, $Y\in H^\Phi$. This shows that $C$ is not order closed.

\smallskip

It follows from what we established above that $C$ is a law-invariant, coherent,
monotone subset of $L^\Phi$ that is norm closed but fails to be order closed.
Consider now the law-invariant, coherent, cash-additive risk measure
$\rho:L^\Phi\to[-\infty,\infty]$ defined by setting
\begin{equation}
\label{eq: cha}
\rho(X) = \inf\big\{m\in\mathbb{R} : X+m1_\Omega\in C\big\}, \ \ \ X\in L^\Phi.
\end{equation}
It is immediate to see that $\rho$ does not attain the value $-\infty$ and is
proper. Moreover, since $C$ is norm closed, $\rho$ is norm lower semicontinuous
by \cite[Corollary~3.3.8]{M:15}. However, since $\{\rho\leq0\}=C$ by
\cite[Proposition~3.2.7]{M:15}, it follows that $\rho$ is not order lower
semicontinuous or, equivalently, fails to have the Fatou property. This proves
that (1) implies (2) and concludes the proof.
\end{proof}

\begin{remark}
The functional defined in \eqref{eq: cha} provides an explicit example of a
law-invariant, coherent, cash-additive risk measure on $L^\Phi$ that is norm lower
semicontinuous but fails to satisfy the Fatou property when $\Phi$ is not $\Delta_2$.
\end{remark}

\smallskip

The last two results announced in the introduction, namely the generalization of
Kusuoka's representation and the Fatou-property-preserving extension result,
will be derived from the following ``localization'' lemma. This result has an
independent interest once we recall that, working in a general Orlicz space, the
space $L^\infty$ need not be norm dense in $L^\Phi$.

\begin{lemma}\label{uniq}
Let $\rho_1,\rho_2:L^\Phi\to(-\infty,\infty]$ be proper, quasiconvex, law-invariant
functionals with the Fatou property. Then, we have $\rho_1=\rho_2$ whenever
$\rho_1$ and $\rho_2$ coincide on $L^\infty$.
\end{lemma}
\begin{proof}
Fix any $X\in L^\Phi$ and, by Proposition~\ref{ce-con}, take a sequence
$(\pi_n)\subset\Pi$ such that
$\mathbb{E}[X|\pi_n]\xrightarrow{o}X$. Since $\rho_1$ is order lower
semicontinuous, we have
\[
\rho_1(X) \leq \liminf_{n\to\infty} \rho_1(\mathbb{E}[X|\pi_n]).
\]
The set $C=\big\{Y\in L^\Phi:\rho_1(Y)\leq \rho_1(X)\big\}$ is convex,
law-invariant, order closed and clearly contains $X$. Hence, by
Corollary~\ref{cha-order}, we have $\mathbb{E}[X|\pi_n]\in C$ for every
$n\in\N$, so that
\[
\limsup_{n\to\infty}\rho_1(\mathbb{E}[X|\pi_n]) \leq \rho_1(X).
\]
As a consequence, we infer that $\rho_1(\mathbb{E}[X|\pi_n])\to\rho_1(X)$. The
same conclusion holds for $\rho_2$ as well. Since $\mathbb{E}[X|\pi_n]\in
L^\infty$ for every $n\in\N$, it follows from our assumption that
$\rho_1(X)=\rho_2(X)$.
\end{proof}

\smallskip

\noindent
The following result will be also needed, in addition to the preceding lemma, to
establish the generalization of Kusuoka's representation.

\begin{lemma}\label{rest}
Let $\rho:L^\Phi\rightarrow(-\infty,\infty]$ be a proper, quasiconvex, law-invariant
functional with the Fatou property. Then, its restriction to $L^\infty$ is also
proper and has the Fatou property.
\end{lemma}
\begin{proof}
Denote by $\rho_{|L^\infty}$ the restriction of $\rho$ to $L^\infty$ and take
any $X_0\in L^\Phi$ such that $\rho(X_0)<\infty$. Then, since the sublevel set
$C=\big\{Y\in L^\Phi:\rho(Y)\leq \rho(X_0)\big\}$ is convex, law-invariant,
order closed and contains $X_0$, it follows from Corollary~\ref{cha-order} that
$\mathbb{E}[X_0]1_\Omega\in C$, so that $\rho(\mathbb{E}[X_0]1_\Omega)<\infty$. This
proves that $\rho_{|L^\infty}$ is proper. Take now a sequence
$(X_n)\subset L^\Phi$ and $X\in L^\Phi$ such that $X_n\to X$ a.s. and $\abs{X_n}\leq Y$ for some $Y\in L^\infty$ and
all $n\in\N$. Since $Y\in L^\Phi$, it follows that $X_n\xrightarrow{o}X$ in
$L^\Phi$. Thus, we infer that
$\rho_{|L^\infty}(X)\leq\liminf_{n\to\infty}\rho_{|L^\infty}(X_n)$ and this
shows that $\rho_{|L^\infty}$ has the Fatou property.
\end{proof}

\smallskip

\begin{proof}[{\bf{Proof of Theorem~\ref{KT}}}]
Denote by $\rho_{|L^\infty}$ the restriction of $\rho$ to $L^\infty$. It follows
from Lemma~\ref{rest} that $\rho_{|L^\infty}$ is a convex, law-invariant,
cash-additive risk measure satisfying the Fatou property. Then,
\cite[Theorem~4.62]{FS:04} implies that
\[
\rho_{|L^\infty}(X) =
\sup_{\mu\in\mathcal{P}((0,1])}\Big(\int_{(0,1]}
\ES_\alpha(X)\,\mathrm{d}\mu(\alpha)-\gamma(\mu)\Big), \ \ \ X\in L^\infty,
\]
where
\[
\gamma(\mu) = \sup_{X\in L^\infty, \,\rho(X)\leq
0}\int_{(0,1]}\ES_\alpha(X)\,\mathrm{d}\mu(\alpha), \ \ \ \mu\in
\mathcal{P}((0,1]).
\]
Now, define $\rho':L^\Phi\to(-\infty,\infty]$ by setting
\[
\rho'(X) =
\sup_{\mu\in\mathcal{P}((0,1])}\Big(\int_{(0,1]}
\ES_\alpha(X)\,\mathrm{d}\mu(\alpha)-\gamma(\mu)\Big), \ \ \ X\in L^\Phi.
\]
Clearly, $\rho'$ is a convex, law-invariant, cash-additive risk measure
satisfying the Fatou property. Moreover, since $\rho$ and $\rho'$ coincide on
$L^\infty$, it follows that $\rho=\rho'$ by Lemma~\ref{uniq}. This shows that
$\rho$ has the desired representation.
\end{proof}

\smallskip

\begin{proof}[{\bf{Proof of Theorem~\ref{ext}}}]
By \cite[Theorem~2.2]{FS:12}, $\rho $ admits a proper, convex, law-invariant
extension $\overline{\rho}$ to $L^1$ that is norm lower semicontinuous. Denote by $\rho'$ the restriction of \,$\overline{\rho}$ to $L^\Phi$ and note that $\rho'$ has the Fatou property. To see this, consider a sequence $(X_n)\subset L^\Phi$ and $X\in L^\Phi$ such that $X_n\to X$ a.s.~and $|X_n|\leq Y$ for some $Y\in L^\Phi$ and all
$n\geq 1$. Since $Y\in L^1$, the Dominated Convergence Theorem implies that $X_n\xrightarrow{\norm{\cdot}_1}X$ and therefore $\rho'(X)\leq\liminf_{n\to\infty}\rho'(X_n)$ by norm
lower semicontinuity of \,$\overline{\rho}$.
The uniqueness
follows from Lemma~\ref{uniq}.
\end{proof}

\begin{remark}
Differently from the case of bounded positions, Theorem 3 and Theorem 4 no longer hold if we replace the Fatou property, or equivalently order lower semicontinuity, by norm lower semicontinuity. Indeed, assume that
$\Phi$ is not $\Delta_2$ and let $\rho$ be the coherent, law-invariant, norm
lower semicontinuous, cash-additive risk measure constructed in~\eqref{eq: cha}.
Then, $\rho$ does not admit a Kusuoka-type representation because it would otherwise satisfy the Fatou property.
Moreover, note that, being cash-additive, the restriction of $\rho$ to $L^\infty$, denoted by
$\rho_{|L^\infty}$, is norm continuous.
Applying Theorem~\ref{ext}, we obtain a convex, law-invariant extension $\rho'$
of $\rho_{|L^\infty}$ to the whole of $L^\Phi$ that satisfies the Fatou
property. Now, $\rho$ and $\rho'$ coincide on $L^\infty$ but are not equal since
one of them has the Fatou property whereas the other does not.
\end{remark}


\end{document}